\documentclass[10pt,transaction,twocolumn]{IEEEtran}
\usepackage{amsmath,amsfonts,amssymb,mathrsfs,amsthm}
\usepackage{latexsym}
\usepackage{stmaryrd}
\usepackage[nospace,noadjust]{cite}
\usepackage{latexsym}
\usepackage{epsfig}
\usepackage{epstopdf}
\usepackage[nospace,noadjust]{cite}
\usepackage{array}
\usepackage[english]{babel}
\usepackage{color}
\usepackage{babel}
\usepackage{graphicx,xcolor,colortbl}    
\usepackage{epstopdf,epsfig}
\usepackage{graphicx}
\usepackage{tabularx,ragged2e,booktabs,caption}
\usepackage{algorithm,algpseudocode}
\usepackage{tikz}
\usepackage{graphicx}
\usepackage{caption}
\usepackage{hhline}
\usepackage{caption}
\usepackage{subcaption}
\usepackage{multirow}
\usetikzlibrary{shapes,arrows}
\usepackage{xcolor,colortbl}
\newtheorem{theorem}{{\bf Theorem}}
\newtheorem{proposition}{\noindent {\bf Proposition}}

\include{new_commands}
\newcommand{\tr}{\mathop{\rm tr}}
\newcommand{\diag}{\mathop{\rm diag}}

\makeatletter
\let\l@ENGLISH\l@english
\makeatother

\begin{document}%
\title{Numerically Stable Evaluation of Moments of Random Gram Matrices with Applications }
\author{Khalil Elkhalil,~\IEEEmembership{Student Member,~IEEE,} Abla Kammoun,~\IEEEmembership{Member,~IEEE}, Tareq~Y.~Al-Naffouri,~\IEEEmembership{Member,~IEEE}, and Mohamed-Slim Alouini,~\IEEEmembership{Fellow,~IEEE}

\thanks{
Copyright (c) 2016 IEEE. Personal use of this material is permitted. However, permission to use this material for any other purposes must be obtained from the IEEE by sending a request to pubs-permissions\@ieee.org.
K. Elkhalil, A. Kammoun, T. Y. Al-Naffouri and M.-S. Alouini are with the Electrical Engineering Program, King Abdullah University of Science and Technology, Thuwal, Saudi Arabia; e-mails: \{khalil.elkhalil, abla.kammoun, tareq.alnaffouri, slim.alouini\}@kaust.edu.sa.
}

}

\maketitle
\begin{abstract}
    This paper is  focuses on the computation of the positive moments of one-side correlated random Gram matrices. Closed-form expressions for the moments can be obtained easily, but numerical evaluation thereof is prone to numerical stability,  especially in high-dimensional settings. This letter provides a numerically stable method that efficiently computes the positive moments in closed-form. The developed expressions are more accurate and can lead to higher
    accuracy levels when fed to moment based-approaches.  As an application, we show how the obtained moments can be used to approximate the marginal distribution of the eigenvalues of random Gram matrices. 
    
\end{abstract}
\begin{IEEEkeywords}
 Gram matrices, one sided correlation, positive moments, Laguerre polynomials.
\end{IEEEkeywords}
\section{Introduction}
\PARstart{G}{ram} random matrices with one-sided correlation naturally arise in the context of signal processing \cite{khalil_tsp} and wireless communications \cite{giusi,mckay_pdf}. For instance, in signal processing, inverse moments of this kind of matrices are used to evaluate the performance of linear estimators such as the best linear unbiased estimator (BLUE) and optimize the design of some covariance matrix estimators  \cite{khalil_tsp,MITthesis}. In wireless
communications, Gram random matrices arise as a key element in the computation of the ergodic capacity of amplify and forward (AF) multiple input multiple output (MIMO) dual-hop systems \cite{mckay_pdf}.  \par 
Given a random matrix $\mathbf{H} \in \mathbb{C}^{q \times n_t}$ such that $\mathbf{H}=\mathbf{\Lambda}^{\frac{1}{2}}\mathbf{X}$, where $\mathbf{X}$ is  a standard complex Gaussian random matrix and $\mathbf{\Lambda}$ is a Hermitian positive definite matrix, the authors in \cite{mckay_pdf} derive a closed form expression of the marginal probability density function (PDF) of an unordered eigenvalue of the Gram matrix $\mathbf{W}=\mathbf{H}^{H}\mathbf{H}$ for arbitrary dimensions
$n_t$ and
$q$. The result, although being very useful, involves the inversion of a large Vandermonde matrix and as such might not be always  numerically stable, especially in the following situations:  
\begin{itemize}
	\item The dimensions $q$ belongs to moderate to large values. 
	\item The gap between the eigenvalues of $\mathbf{\Lambda}$ is small. 
\end{itemize}
To solve this problem, an expression of the exact marginal PDF has been proposed in \cite{moe_win_09} when some eigenvalues are identical. However, the problem of numerical stability remains  when some eigenvalues of $\boldsymbol{\Lambda}$ are different but close to each other.
Motivated by these facts, we provide a more stable method to compute the positive moments of $\mathbf{W}$ without the need to invert large Vandermonde matrices. The contributions of this letter are summarized as follows:
\begin{itemize}
	\item We provide a numerically stable method to evaluate the positive moments of $\mathbf{W}$.
    \item Using Laguerre polynomials and based on the calculated positive moments, we provide a numerically stable approximation of the marginal probability density function (PDF) of the eigenvalues  of $\mathbf{W}$. 
\end{itemize}

\par 
The remainder of this letter is organized as follows. In Section \ref{method}, we provide the main steps to efficiently compute the positive moments of $\mathbf{W}$. In Section \ref{PDF approximation}, we propose to approximate the marginal PDF based on the computed positive moments and using Laguerre polynomials. In Section \ref{results}, we present some numerical results to validate our method and finally we conclude our work in Section \ref{conclusion}.
%
\section{A numerically stable method to compute the moments of random Gram matrices}
\label{method}
\subsection{Problem statement}
Let $\mathbf{H} =\boldsymbol{\Lambda}^{\frac{1}{2}} {\bf X}\in\mathbb{C}^{q\times n_t}$ where $\mathbf{\Lambda} \in \mathbb{C}^{q\times q}$ is a positive definite matrix with distinct eigenvalues $0<\beta_1< \beta_2< \cdots <\beta_q$ and ${\bf X}$ a standard complex Gaussian matrix. Assume that $n_t\leq q$.  The marginal PDF of an unordered eigenvalue $\lambda$ of ${\bf W}={\bf H}^{H}{\bf H}$ is given by \cite[Lemma 1]{mckay_pdf} 
\begin{align}
\label{pdf}
	f_{\lambda}\left(\lambda\right) & = \frac{1}{n_t \prod_{i<j}^q\left(\beta_j-\beta_i\right)} \nonumber \\ & \times \sum_{l=1}^{q}\sum_{k=q-n_t+1}^{q} \frac{\lambda^{n_t+k-q-1}e^{-\lambda/\beta_l}\beta_l^{q-n_t-1}}{\Gamma\left(n_t-q+k\right)}D_{l,k},
\end{align}
where  $D_{l,k} = \{\mathbf{D}\}_{l,k}$ is the $\left(l,k\right)th$ cofactor of the Vandermonde $q \times q$ matrix $\mathbf{\Psi}$ whose $\left(m,n\right)$th entry is 
\begin{equation}
\{\mathbf{\Psi}\}_{m,n} = \beta_m^{n-1}.
\end{equation}
Expressing the inverse of $\boldsymbol{\Psi}$ as
\begin{align*}
\mathbf{\Psi}^{-1} & = \frac{1}{\text{det}\left(\mathbf{\Psi}\right)} \mathbf{D}^T \\
& = \frac{1}{ \prod_{i<j}^q\left(\beta_j-\beta_i\right)}  \mathbf{D}^T,
\end{align*}
 the PDF in (\ref{pdf})  simplifies to
\begin{align}
\label{pdf1}
f_{\lambda}\left(\lambda\right) = \frac{1}{n_t }  \sum_{l=1}^{q}\sum_{k=q-n_t+1}^{q} \frac{\lambda^{n_t+k-q-1}e^{-\lambda/\beta_l}\beta_l^{q-n_t-1}}{\Gamma\left(n_t-q+k\right)}\mathbf{\Psi}^{-1}_{k,l}.
\end{align}
The cumulative density function (CDF) can thus be easily derived as
\begin{equation}
	\label{cdf_mckay}
	F_{\lambda}\left(\lambda\right) = \frac{1}{n_t }  \sum_{l=1}^{q}\sum_{k=q-n_t+1}^{q} \frac{\beta_l^{k-1}\gamma\left(n_t-q+k,\lambda / \beta_l\right)}{\Gamma\left(n_t-q+k\right)}\mathbf{\Psi}^{-1}_{k,l},
\end{equation}
where $\Gamma\left(.\right)$ and $\gamma\left(.,.\right)$ are respectively the standard Gamma and the lower incomplete Gamma functions. \par
Knowing the marginal PDF, it is possible to compute the expected value of any functional $g$ of the eigenvalues of ${\bf W}$. Indeed, we have
\begin{align}
    \mathbb{E}\left[g({\bf W})\right]&\triangleq \frac{1}{n_t}\sum_{i=1}^{n_t}g(\lambda_i({\bf W}))\nonumber\\
    &= \int_{0}^{\infty} g(\lambda)f_\lambda(\lambda)d\lambda,
\label{eq:g}
\end{align}
where $\left\{\lambda_i({\bf W})\right\}_{i=1}^{n_t}$ are the eigenvalues of ${\bf W}$. 
Equation \eqref{eq:g} is very useful in practice as it can be leveraged to  compute performance metrics of many wireless communication and signal processing schemes. Examples include the ergodic capacity,  the SINR at the output of the MMSE receiver and the MSE of the BLUE estimator which correspond respectively to selecting $g(x)$ as $g(x)=\log_2(x+\sigma^2)$, $g(x)=\frac{1}{x+\sigma^2}$ where $\sigma^2$ is the noise variance and $g(x)=x^{-1}$. 

When it comes to numerically compute $\mathbb{E}\left[g({\bf W})\right]$, it is easy to see that, when the eigenvalues  $\left\{\beta_i\right\}$ are very close causing the matrix $\boldsymbol{\Psi}$ to be ill-conditioned, some numerical stability issues might occur.  
In this work, we show that for some functionals $g$, namely polynomials, it is possible to  evaluate  $ \mathbb{E}\left[g({\bf W})\right]$  in a stable way. This allows us, using moment approximation techniques, to obtain a numerically stable approximation of the marginal PDF. 
We believe that  the same approximation method can  also be extended to approximate $\mathbb{E}\left[g({\bf W})\right]$ for any functional $g$ of interest. 
\subsection{A Numerically stable method to compute positive moments}
In this section, we propose a numerically stable technique to compute the positive moments of ${\bf W}$. Let $p\in\mathbb{N}$, the $p$-th moment of matrix ${\bf W}$ is given by
\begin{equation}
\label{moment}
\begin{split}
\mu_{\mathbf{W}}\left(p\right)& = \mathbb{E}\left[\lambda^p\right]\\ & = \frac{1}{n_t} \mathbb{E}\:\tr \left[\mathbf{W}^p\right] \\ & =\int_{0}^{\infty}\lambda^p f_{\lambda}\left(\lambda\right) \\
& = \frac{1}{n_t} \sum_{k=q-n_t+1}^{q}\frac{\Gamma\left(n_t+p+k-q\right)}{\Gamma\left(n_t-q+k\right)}\sum_{l=1}^{q}\mathbf{\Psi}^{-1}_{k,l}\beta_l^{p+k-1}.
\end{split}
\end{equation}
In many practical scenarios, numerical instability might originate from the computation of the following quantity
$$
\sum_{l=1}^{q}\mathbf{\Psi}^{-1}_{k,l}\beta_l^{p+k-1},
$$
as $\mathbf{\Psi}$ is ill-conditioned. To overcome this issue, we propose an alternative way that avoids computing the inverse of $\boldsymbol{\Psi}$. For $k\in\llbracket 1,q\rrbracket$ and $\tau\in \llbracket p+q-n_t,p+q-1 \rrbracket$, define $\alpha_{k,\tau}$ as
$$
\alpha_{k,\tau}=\sum_{l=1}^{q}\mathbf{\Psi}^{-1}_{k,l}\beta_l^{\tau}.
$$
The basic idea is based on the observation that $\boldsymbol{\alpha}_\tau\triangleq\left[\alpha_{1,\tau},\cdots,\alpha_{q,\tau}\right]^{T}$ is solution to the following linear system
\begin{equation}
\label{system}
\mathbf{\Psi} \boldsymbol{\alpha}_\tau = \boldsymbol{\beta}_\tau.
\end{equation}
where $\boldsymbol{\beta}_\tau=\left[\beta_1^{\tau},\cdots,\beta_{q}^{\tau}\right]^{T}$.
If $0\leq \tau \leq q$, then a straightforward solution to \eqref{system} is given by $\boldsymbol{\alpha}_\tau=\left[{\bf 0}_{\tau\times 1}^{T},1,{\bf 0}_{q-\tau\times 1}^{T}\right]^{T}$. From now on, we assume that $\tau > q$.  

Writing \eqref{system} in the following equivalent way
$$
\beta_k^{\tau} =\sum_{l=1}^q \beta_k^l \alpha_{l,\tau},\:\: k=1,\cdots,q
$$
we can easily see that $\left\{\beta_i\right\}_{i=1}^q$  are roots of the following polynomial:
\begin{equation}
	\label{polynom}
    P\left(X\right) = \sum_{k=1}^{q}\alpha_{k,\tau} X^{k-1} - X^{\tau}.
\end{equation}
Hence, there exists $Q(X)$ a polynomial with degree $\tau-q$ such that:
\begin{equation}
	P\left(X\right) = Q\left(X\right)\prod_{i=1}^q \left(X-\beta_i\right),
\end{equation}
Note that exact knowledge of $P(X)$ leads to the determination of the unknown coefficients $\alpha_{k,\tau}$, since they are by construction  among the coefficients of $P(X)$. To fully characterize $P(X)$, we first observe that 
\begin{itemize}
    \item  the coefficients of $P$ associated  with exponents $X^q, X^{q+1}, \cdots,X^{\tau-1}$ are all zero.
    \item the coefficient associated with $X^{\tau}=-1$. 
    \end{itemize}
    Let   $\left\{a_i\right\}_{i=1}^{q+1}$ be the coefficients of $\prod_{i=1}^q \left(X-\beta_i\right)$ (i.e, $\prod_{i=1}^q
\left(X-\beta_i\right)=\sum_{i=1}^{q+1}a_iX^{i-1}$), which can be exactly obtained using the Newton-Girard algorithm \cite{girard_newton}. Let
$\left\{b_i\right\}_{i=1}^{\tau-q+1}$ be the coefficients of $Q(X)$ so that:
$	Q\left(X\right) = \sum_{k=1}^{\tau - q+1}b_k X^{k-1}$. From the available information about the coefficients of $P$, we can show that 
$\left\{b_k\right\}_{k=1}^{\tau-q+1}$  satisfy the following set of equations
\begin{equation}
\begin{cases}
    &	a_{q+1}b_1 + a_qb_2 + \cdots a_{2q+1-\tau}b_{\tau-q+1} = 0 \\ 
    &a_{q+1}b_2 + a_qb_3 + \cdots a_{2q+2-\tau}b_{\tau-q+1} = 0 \\
    &\hspace{2cm}\:\:\:\:
\: \vdots\\ 
&a_{q+1}b_{\tau-q+1} = -1. 
\end{cases}
\label{eq:sys}
\end{equation}
where we use the convention that $a_j=0$ if $j\leq 0$ or $j>q+1$. 
The  system of equation in \eqref{eq:sys} can be also expressed in the following matrix form:
\begin{equation} \label{gaussian_system}
	\mathbf{\Phi} \begin{bmatrix}
	b_1\\b_2 
	\\ \vdots
	\\ b_{\tau-q+1}
	
	\end{bmatrix} =\begin{bmatrix}
	0\\0
	\\ \vdots
	\\0 \\ -1
	\end{bmatrix} ,
\end{equation}
where $\mathbf{\Phi}$ is the upper triangular matrix given by
\begin{equation}
	\mathbf{\Phi} = \begin{bmatrix}
        a_{q+1} &a_q  &\cdots  & &\cdots a_{2q+1-\tau}  \\ 
        0&a_{q+1}  & a_q & \cdots &\cdots a_{2q+2-\tau}    \\ 
        \vdots&  \ddots&  \ddots & &\vdots\\ 
        \vdots&  & \ddots & \ddots &a_q  \\ 
        0&\cdots  && 0 & a_{q+1}  
	\end{bmatrix}.
\end{equation}
Vector ${\bf b}=\left[b_1,\cdots,b_{\tau-q+1}\right]^{T}$ can be thus determined by taking the inverse of matrix $\boldsymbol{\Phi}$ as
 \begin{equation}
 \begin{bmatrix}
 b_1\\b_2 
 \\ \vdots
 \\ b_{\tau-q+1}
 
 \end{bmatrix} = \mathbf{\Phi} \setminus \begin{bmatrix}
 0\\0
 \\ \vdots
 \\0 \\ -1
 
 \end{bmatrix}.
 \end{equation}
 From a numerical standpoint, this operation, involving inversion of an upper triangular matrix,  can be solved in a stable fashion using back-substitution algorithm and is, as such, much more stable than the inversion of matrix $\boldsymbol{\Psi}$  required in the evaluation of \eqref{moment}. Once coefficients $\left\{b_i\right\}_{i=1}^{\tau-q+1}$ are obtained, $\left\{\alpha_{k,\tau}\right\}$ can be evaluated as \footnote{This can be seen by using the fact that
 $P(X)=\sum_{i=1}^{q+1}\sum_{k=1}^{\tau-q+1}a_ib_k X^{i+k-2}$.}
 $$
 \alpha_{j,\tau}=\sum_{k=1}^{\tau-q+1} b_ka_{j+2-k}. 
$$

To validate our procedure, we compute the positive moments of the Gram matrix $\mathbf{W}$ in the case where the correlation matrix $\mathbf{\Lambda}$ follows the following model \cite{khalil_tsp}
 \begin{equation}
 \label{toeplitz}
 \mathbf{\Lambda}= \left(1-\xi\right)\diag\left(1,\xi,\xi^2,\cdots,\xi^{q-1}\right) , \:\: 0 \leq \xi \leq 1,
 \end{equation}
 where the coefficient $\xi$ indicates the forgetting factor. This kind of matrices arise in covariance matrix estimation and more precisely in exponentially weighted sample covariance matrix (more details can be found in \cite{khalil_tsp}, section III-B)
 . Note that for moderate to large values of $q$, the eigenvalues of $\mathbf{\Lambda}$ given by $\left(1-\xi\right), \left(1-\xi\right)\xi, \cdots, \left(1-\xi\right)\xi^{q-1}$ are very close to each other,  which might cause singularity issues when using the formula in (\ref{pdf}). We consider two different configurations $\text{config}_1$ and $\text{config}_2$ corresponding respectively to $(n_t=3, q=5)$ and $(n_t=3, q=20)$. For both configurations, we evaluate the moments using 
 \eqref{pdf} and the proposed method. We compare the obtained moments with the empirical ones evaluated over $10^{6}$ realizations. 
 
 The results are summarized in Table \ref{moments_table}. 
 As a first observation, we notice that our method provides very close results to the empirical moments while the evaluation of the moments using \eqref{pdf} becomes totally inaccurate in configuration $\text{config}_2$ associated with a higher $q$.    
 This clearly demonstrates the efficiency and the accuracy of our method in calculating the positive moments.

\begin{table}[]
	\centering
	\caption{Positive moments evaluation of the Gram matrix $\mathbf{W}$ with $\mathbf{\Lambda}$ as in (\ref{toeplitz}) with $\xi=0.85$. Both settings are considered: $\text{config}_1: n_t=3, q=5$ and $\text{config}_2: n_t=3, q=20$ for different moment order values, $p$.}
	\label{moments_table}
	\begin{tabular}{c|c|c|c|}
		\cline{2-4}
		&   Formula in \cite{mckay_pdf}    & Empirical($10^6$ realizations) & Proposed  \\ \hline
		\multicolumn{1}{|c|}{$\text{config}_1$, $p=1$} & 0.5563          & 0.5563    & 0.5562    \\ \hline
		\multicolumn{1}{|c|}{\cellcolor[HTML]{C0C0C0} $\text{config}_2$, $p=1$} &\cellcolor[HTML]{C0C0C0} -2.1177e+06 & 0.9613    & 0.9612        \\ \hline
		\multicolumn{1}{|c|}{$\text{config}_1$, $p=5$} & 1.1029    & 1.1031   & 1.1032   \\ \hline
		\multicolumn{1}{|c|}{\cellcolor[HTML]{C0C0C0}$\text{config}_2$, $p=5$} &\cellcolor[HTML]{C0C0C0} 7.7212e+03 & 4.7575    & 4.7562    \\ \hline
		\multicolumn{1}{|c|}{$\text{config}_1$, $p=8$} & 5.3799  & 5.4332 & 5.3989 \\ \hline
		\multicolumn{1}{|c|}{\cellcolor[HTML]{C0C0C0}$\text{config}_2$, $p=8$} & \cellcolor[HTML]{C0C0C0} 5.4568e+04 & 37.3178   & 37.47   \\ \hline
	\end{tabular}
\end{table}
%
 \section{Moment-Based Approach for Density Approximation}\label{PDF approximation}
 In this section, we show that  the knowledge of all positive moments $\mu_{\mathbf{W}}\left(k\right)$, $k=1,2,\cdots$ can be leveraged to approximate the PDF $f_{\lambda}$. In general, retrieving a positive PDF from the knowledge of all its moments is  known as the Stieltjes moment  \cite{rao_book,krein} problem. We say that a PDF is called M-determinate if it can be uniquely determined by its  moments. A sufficient condition for a PDF to be M-determinate is given by the Krein and the Lin
 conditions summarized below     
 \begin{theorem}\cite{krein}
     Let $f$ be a distribution defined in the real half-line $\left(0,\infty\right)$.  	If the following conditions are satisfied:
 	\begin{enumerate}
 		\item The Krein condition:
 	\begin{equation}
        \int_0^{\infty} \frac{-\log f\left(\lambda^2\right)}{1+\lambda^2}d\lambda=\infty
 	\end{equation}
 	\item The Lin condition: $f$ is differentiable and 
 	\begin{equation}
        \lim_{\lambda \rightarrow \infty}\frac{-\lambda \frac{\partial f\left(\lambda\right)}{\partial \lambda}}{f\left(\lambda\right)}= \infty
 	\end{equation}
 \end{enumerate}
 Then, $f_{\lambda}$ is M-determinate.
 \end{theorem}  
 \begin{proposition} \label{proposition1}
 	The PDF in (\ref{pdf1}) is M-determinate.
 \end{proposition}
\begin{proof}
	See Appendix for a proof.
\end{proof}
Now that we prove that $f_{\lambda}$ is M-determinate, an approximation of the marginal density, involving laguerre polynomials,  can be derived as \cite{provost}: 
 \begin{equation} \label{apprx_provost}
 f_{\lambda}\left(\lambda\right) =\frac{ \lambda^{\nu} e^{-\lambda / c}}{c^{\nu+1}} \sum_{i=0}^{\infty} \delta_i \mathcal{L}_i\left(\nu,\lambda / c\right),
 \end{equation}
 where $c= \frac{\mu_{\mathbf{W}}\left(2\right)-\mu_{\mathbf{W}}\left(1\right)^2}{\mu_{\mathbf{W}}\left(1\right)}$, $\nu = \frac{\mu_{\mathbf{W}}\left(1\right)}{c} -1$,
 \begin{equation}
 \mathcal{L}_i\left(\nu,\lambda\right) = \sum_{k=0}^{i} \left(-1\right)^k \frac{\Gamma\left(\nu+i+1\right)\lambda^{i-k}}{k! \left(i-k\right)!\: \Gamma\left(\nu+i-k+1\right)},
 \end{equation}
 is the \emph{Laguerre} polynomial of order $i$ in $\lambda$ and parameter $\nu$ and 
 \begin{equation}
 \delta_i = \sum_{k=0}^{i} \frac{\left(-1\right)^k}{c^{i-k}}  \frac{i!}{k! \left(i-k\right)! \: \Gamma\left(\nu+i-k+1\right)}\mu_{\mathbf{W}}\left(i-k\right).
 \end{equation}
 Truncating the series in \eqref{apprx_provost} at order $K$ yields the following approximation for the marginal PDF
 \begin{equation}
f_{\lambda,K}\left(\lambda\right) =\frac{ \lambda^{\nu} e^{-\lambda / c}}{c^{\nu+1}} \sum_{i=0}^{K} \delta_i \mathcal{L}_i\left(\nu,\lambda /c\right).
 \end{equation}
 The CDF can thus be approximated as follows
  \begin{equation}
  \begin{split}
  F_{\lambda,K}\left(\lambda\right) = \sum_{i=0}^{K} \delta_i \sum_{k=0}^{i} & \left(-1\right)^k  \frac{\gamma\left(i+\nu-k+1,\lambda / c\right)}{k!\left(i-k\right)!}.
  \end{split}
  \end{equation}
 \section{Selected numerical Results} \label{results}
In this section, we investigate the accuracy of the proposed PDF and CDF moment-based approach approximation. To this end, we compare them with their empirical counterparts and those evaluated using the results in \cite{mckay_pdf}.
\begin{figure}[h!]
	\centering
	\includegraphics[scale=0.33]{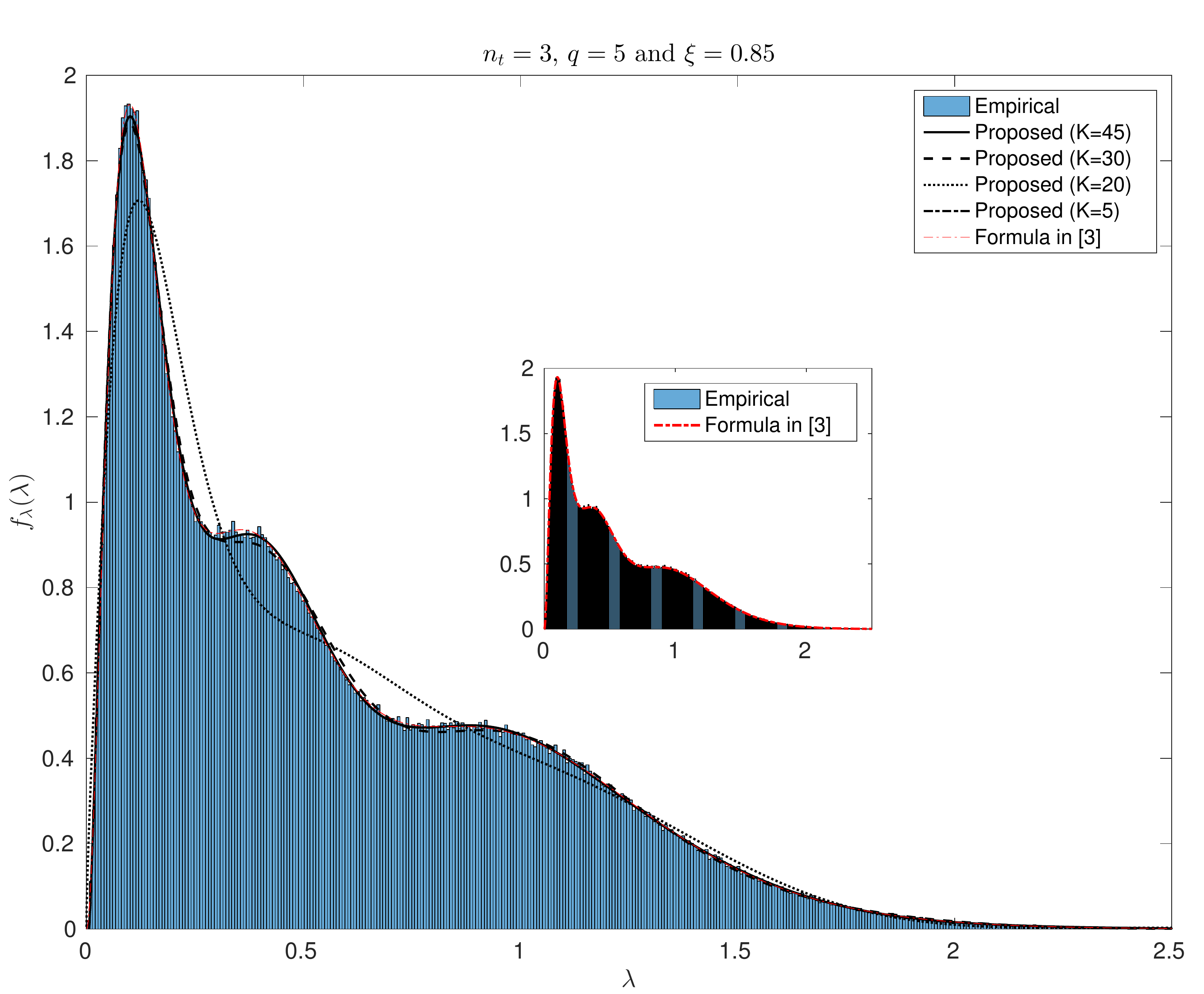}
	\caption{PDF of an unordered eigenvalue of $\mathbf{W}$ with $n_t=3$, $q=5$ and $\xi=0.85$.}
	\label{fig:small_network}
\end{figure}
\begin{figure}[h!]
	\centering
	\includegraphics[scale=0.32]{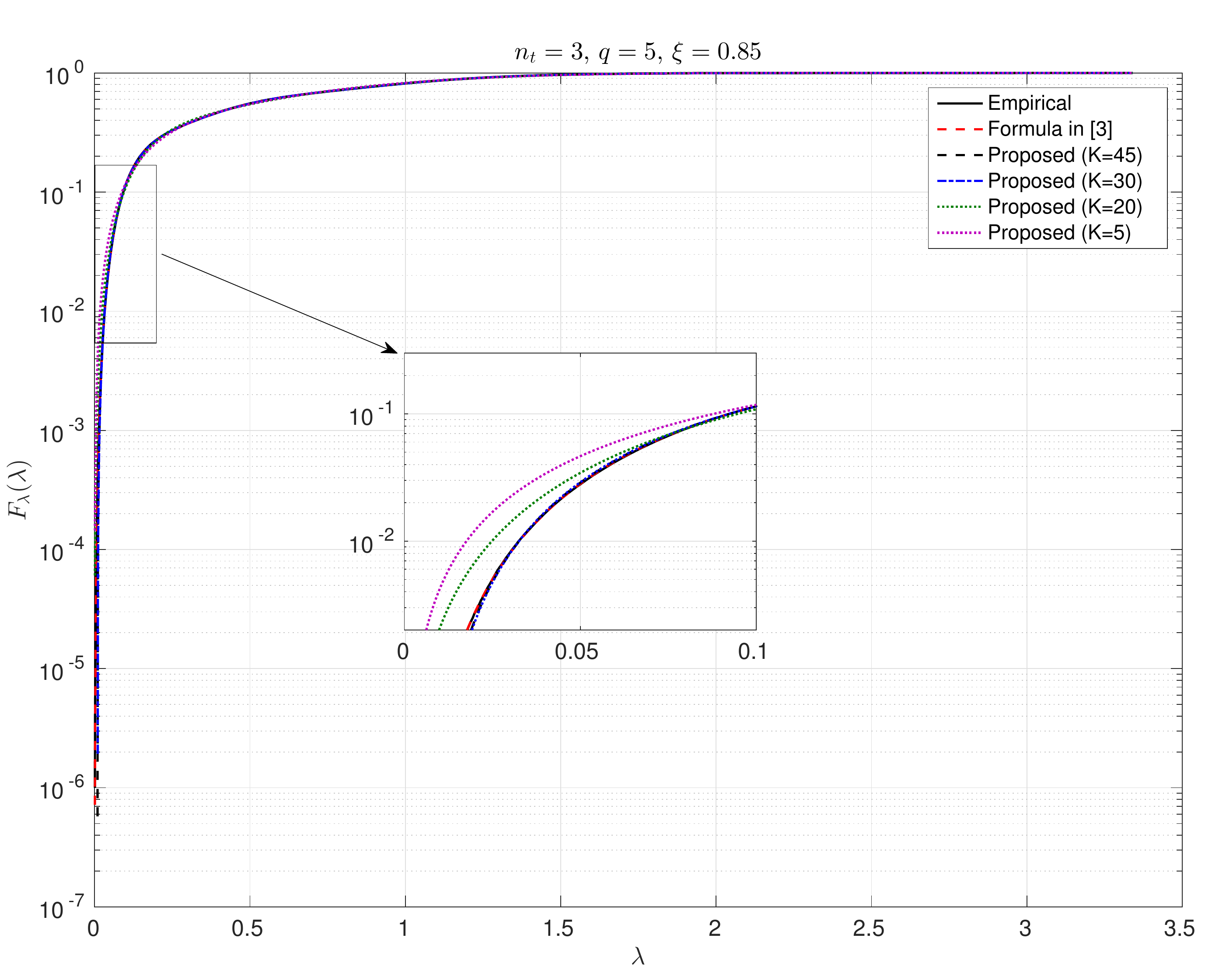}
	\caption{CDF of an unordered eigenvalue of $\mathbf{W}$ with $n_t=3$, $q=5$ and $\xi=0.85$.}
	\label{fig:cdf_small_network}
\end{figure}
\begin{figure}[h!]
	\centering
	\includegraphics[scale=0.31]{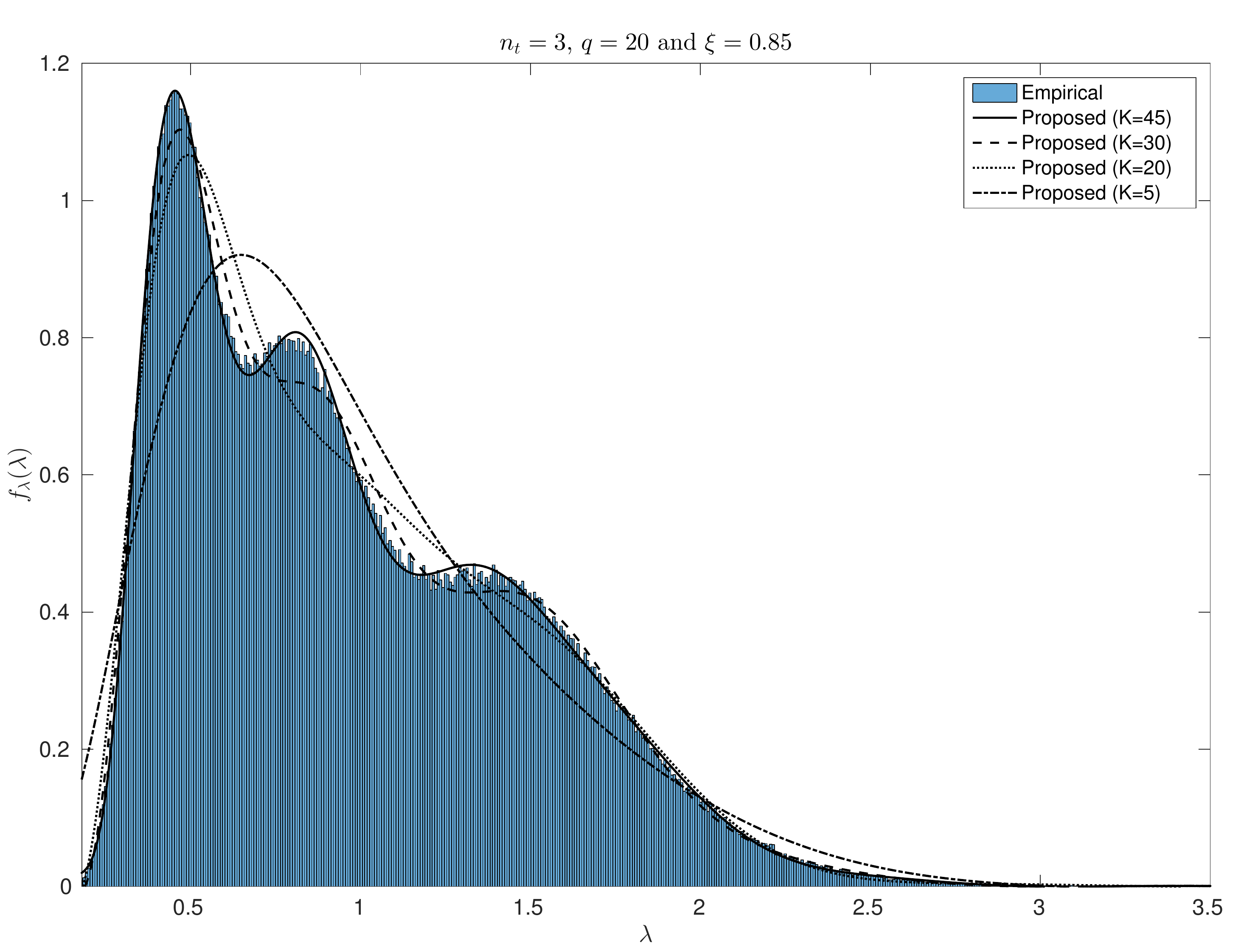}
	\caption{PDF of an unordered eigenvalue of $\mathbf{W}$ with $n_t=3$, $q=20$ and $\xi=0.85$. The plot for the exact formula provided in \cite{mckay_pdf} is omitted due to singularity issues.}
	\label{fig:large_network}
\end{figure}
\begin{figure}[h!]
	\centering
	\includegraphics[scale=0.33]{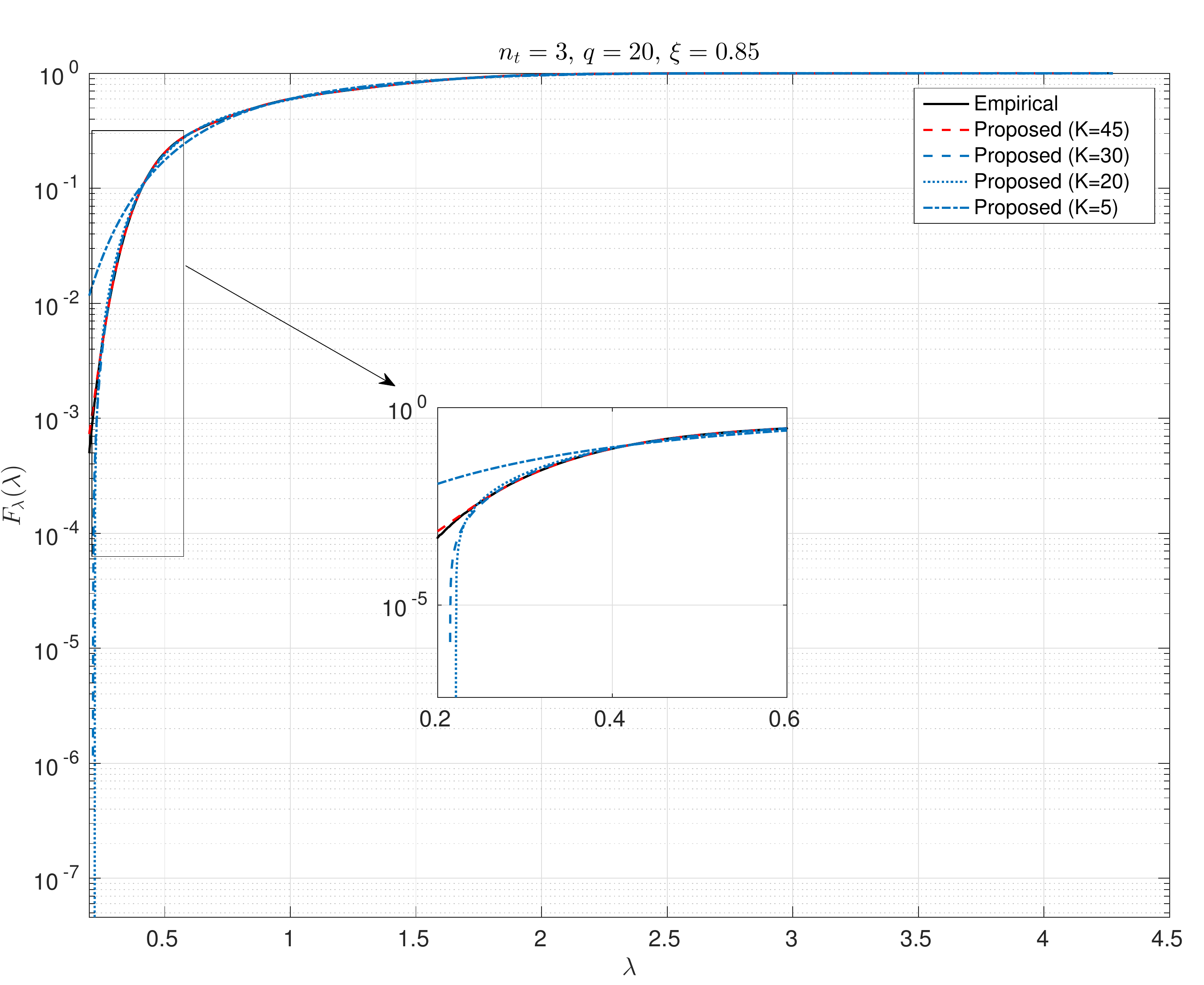}
	\caption{CDF of an unordered eigenvalue of $\mathbf{W}$ with $n_t=3$, $q=20$ and $\xi=0.85$. The plot for the exact formula provided in \cite{mckay_pdf} is omitted due to singularity issues.}
	\label{fig:cdf_large_network}
\end{figure}
\par
In Figure \ref{fig:small_network}, we assume that $\mathbf{\Lambda}$ follows the same model as in (\ref{toeplitz}) with $\xi=0.85$, $n_t=3$ and $q=5$. We compare the accuracy of our approach with the corresponding empirical density and the formula provided in \cite{mckay_pdf}. It can be noticed that our approximation becomes more accurate by increasing the truncation order $K$. As evidenced from Figures \ref{fig:small_network} and \ref{fig:cdf_small_network}, a good approximation
can be achieved starting from $K=30$ for both PDF and CDF. \par 
In Figures \ref{fig:large_network} and \ref{fig:cdf_large_network}, we increase the value of $q$ to $q=20$. In this case, the formula provided in \cite{mckay_pdf} presents severe numerical instability and thus could not be plotted in this case. On the other hand, our moment-based approach achieves a very good approximation starting from $K=30$. For $K=45$, we can see that we have a perfect match with the empirical PDF and CDF. 
 \section{Conclusion} \label{conclusion}
In this paper, we propose a numerically stable method that efficiently compute the positive moments of one-side correlated Gram matrices. From a practical standpoint, these moments can be used to approximate the marginal distribution and CDF of the eigenvalues of Large Gram random matrices and thus constitute an efficient alternative to conventional methods which become highly inaccurate in high dimensional settings. 
\bibliographystyle{IEEEtran}
\bibliography{References}
\section*{Appendix (Proof of Proposition \ref{proposition1})}
\subsection*{Krein condition} We start by rewrite the PDF in (\ref{pdf1}) as 
\begin{equation} \label{modified}
f_{\lambda}(\lambda) = \sum_{l=1}^q\sum_{k=q-n_t+1}^{q} \gamma_{k,l} \lambda^{n_t+k-q-1}e^{-\lambda / \beta_l},
\end{equation}
where $\gamma_{k,l} =  \frac{\beta_l^{q-n_t-1} \mathbf{\Psi}_{k,l}^{-1}}{n_t \:\Gamma(n_t-q+k)}$. Then,
\begin{equation}
\begin{split}
f_{\lambda}(\lambda^2) & = \sum_{l=1}^q\sum_{k=q-s+1}^{q} \gamma_{k,l} \lambda^{2(n_t+k-q-1)}e^{-\lambda^2 / \beta_l} \\
& \leq e^{-\lambda^2 / \beta_q}\sum_{l=1}^q\sum_{k=q-s+1}^{q} \gamma_{k,l} \lambda^{2(n_t+k-q-1)}.
\end{split}
\end{equation}
Thus,
\begin{equation}
-\log\left(f_{\lambda}(\lambda^2)\right) \geq \lambda^2 / \beta_q - \log\left(\sum_{l=1}^q\sum_{k=q-n_t+1}^{q} \gamma_{k,l} \lambda^{2(n_t+k-q-1)}\right).
\end{equation}
and 
\begin{equation}
\begin{split}
\frac{-\log\left(f_{\lambda}(\lambda^2)\right)}{1+\lambda^2} & \geq \frac{\lambda^2}{\beta_q\left(1+\lambda^2\right)} \\
& -\frac{\log\left(\sum_{l=1}^q\sum_{k=q-n_t+1}^{q} \gamma_{k,l} \lambda^{2(n_t+k-q-1)}\right)}{1+\lambda^2}
\end{split}.
\end{equation}
Integrating the first term of the right-hand side term provides infinity while integrating the second term results in a finite value. Thus, the integral diverges to infinity which fulfills the Krein condition. 
\subsection*{Lin condition}
Using the modified expression in (\ref{modified}), we have
\begin{equation}
\begin{split}
\frac{\partial f_{\lambda}(\lambda)}{\partial \lambda} &= \sum_{l=1}^q\sum_{k=q-n_t+1}^{q} \gamma_{k,l} \Biggl[(n_t+k-q-1)\lambda^{n_t+k-q-2}e^{-\lambda /\beta_l}\\
& -1/\beta_l \lambda^{n_t+k-q-1}e^{-\lambda /\beta_l}\Biggr].
\end{split}
\end{equation}
Then,
\begin{equation}
\begin{split}
    &-\lambda \frac{\partial f_{\lambda}(\lambda)}{\partial \lambda}  = \sum_{l=1}^q\sum_{k=q-n_t+1}^{q}\gamma_{k,l}e^{-\lambda /\beta_l} \lambda^{n_t+k-q-1} \Biggl[-(n_t+k-q-1)\\
& +\lambda/\beta_l \Biggr] \\
&\geq -(n_t-1)f_{\lambda}(\lambda) + \frac{\lambda}{\beta_q} f_{\lambda}(\lambda).
\end{split}
\end{equation}
Thus,
\begin{equation}
\begin{split}
\frac{-\lambda \frac{\partial f_{\lambda}(\lambda)}{\partial \lambda}}{f_{\lambda}(\lambda)} & \geq  -\left(n_t-1\right) + \frac{\lambda}{\beta_q} \xrightarrow{\lambda \to \infty} \infty.
\end{split}
\end{equation}
This completes the proof of the proposition.
\end{document}